\documentclass[reprint, 
superscriptaddress,
 amsmath,amssymb,
 aps,
 prab,
floatfix,
showkeys,
longbibliography
]{revtex4-2}
\usepackage{natbib}
\usepackage{dcolumn}
\usepackage{bm}
\usepackage{soul}
\usepackage[utf8]{inputenc} 
\usepackage[T1]{fontenc}    
\usepackage{hyperref}       
\usepackage{url}            
\usepackage{booktabs, arydshln}       
\usepackage{amsfonts}       
\usepackage{nicefrac}       
\usepackage{microtype}      
\usepackage{graphicx}
\usepackage{amsthm}
\usepackage{mathtools}
\usepackage{siunitx}
\usepackage{tikz}
\usepackage{tikzscale}
\usetikzlibrary{positioning,chains}

\definecolor{PSIgray5}{gray}{0.31}
\definecolor{PSIgray4}{gray}{0.41}
\definecolor{PSIgray3}{gray}{0.59}
\definecolor{PSIgray2}{gray}{0.73}
\definecolor{PSIgray1}{gray}{0.90}

\definecolor{PSIred5}{RGB}{197,0,6}
\definecolor{PSIred4}{RGB}{208,71,41}
\definecolor{PSIred3}{RGB}{218,114,82}
\definecolor{PSIred2}{RGB}{231,162,135}
\definecolor{PSIred1}{RGB}{243,210,194}

\definecolor{PSIgreen5}{RGB}{25,116,24}
\definecolor{PSIgreen4}{RGB}{81,138,66}
\definecolor{PSIgreen3}{RGB}{125,165,105}
\definecolor{PSIgreen2}{RGB}{168,195,154}
\definecolor{PSIgreen1}{RGB}{212,226,206}

\definecolor{PSIblue5}{RGB}{0,59,110}
\definecolor{PSIblue4}{RGB}{64,85,131}
\definecolor{PSIblue3}{RGB}{105,118,158}
\definecolor{PSIblue2}{RGB}{152,159,189}
\definecolor{PSIblue1}{RGB}{203,207,223}

\newtheorem{theorem}{Theorem}[section]

\makeatletter
\def\adl@drawiv#1#2#3{%
        \hskip.5\tabcolsep
        \xleaders#3{#2.5\@tempdimb #1{1}#2.5\@tempdimb}%
                #2\z@ plus1fil minus1fil\relax
        \hskip.5\tabcolsep}
\newcommand{\cdashlinelr}[1]{%
  \noalign{\vskip\aboverulesep
           \global\let\@dashdrawstore\adl@draw
           \global\let\adl@draw\adl@drawiv}
  \cdashline{#1}
  \noalign{\global\let\adl@draw\@dashdrawstore
           \vskip\belowrulesep}}
\makeatother

\begin{document}

\title{Forecasting Particle Accelerator
Interruptions\texorpdfstring{\\}{}Using Logistic LASSO Regression}

\author{\href{https://orcid.org/0000-0003-4881-2166}{\includegraphics[scale=0.06]{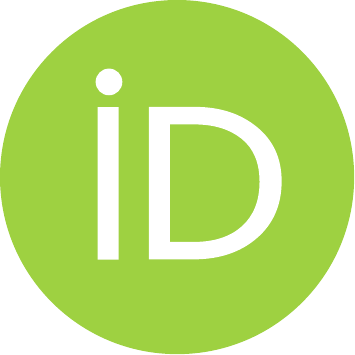}\hspace{1mm}Sichen~Li}}%
\altaffiliation[Also at the ]{Department of Physics, ETH Z\"urich}
\email{sichen.li@psi.ch}
\author{\href{https://orcid.org/0000-0003-4881-2166}{\includegraphics[scale=0.06]{figs/orcid.pdf}\hspace{1mm}
Jochem~Snuverink}}
\affiliation{%
Paul Scherrer Institut, 5232 Villigen, Switzerland}%
\author{\href{https://orcid.org/0000-0003-4881-2166}{\includegraphics[scale=0.06]{figs/orcid.pdf}\hspace{1mm}
Fernando~Perez-Cruz}}
\affiliation{%
Swiss Data Science Center, ETH Z\"urich, Switzerland}%

\author{\href{https://orcid.org/0000-0002-7230-7007}{\includegraphics[scale=0.06]{figs/orcid.pdf}\hspace{1mm}Andreas~Adelmann}}%
\email{andreas.adelmann@psi.ch}
\affiliation{%
Paul Scherrer Institut, 5232 Villigen, Switzerland}%

\date{\today}


\hypersetup{
pdftitle={Forecasting Particle Accelerator Interruptions Using Logistic LASSO Regression},
pdfsubject={},
pdfauthor={Sichen~Li, Jochem~Snuverink, Fernando~Perez-Cruz, Andreas~Adelmann},
pdfkeywords={Machine learning, particle accelerator, beam interruptions, Lasso regression},
}

\begin{abstract}
   Unforeseen particle accelerator interruptions, also known as interlocks, lead to abrupt operational changes despite being necessary safety measures. These may result in substantial loss of beam time and perhaps even equipment damage. We propose a simple yet powerful binary classification model aiming to forecast such interruptions, in the case of the High Intensity Proton Accelerator complex at the Paul Scherrer Institut. The model is formulated as logistic regression penalized by least absolute shrinkage and selection operator, based on a statistical two sample test to distinguish between unstable and stable states of the accelerator.
   
   The primary objective for receiving alarms prior to interlocks is to allow for countermeasures and reduce beam time loss. Hence, a continuous evaluation metric is developed to measure the saved beam time in any period, given the assumption that interlocks could be circumvented by reducing the beam current. The best-performing interlock-to-stable classifier can potentially increase the beam time by around \SI{5}{min} in a day.\ Possible instrumentation for fast adjustment of the beam current is also listed and discussed.
\end{abstract}

\keywords{machine learning, particle accelerator, beam interruptions, LASSO regression}

\maketitle

\section{Introduction}
With the progressive development in data collection, storage and analysis capacities in recent years, data-driven algorithms, especially machine learning (ML) methods, have gained increasing exposure and importance across academia, industry and social life.
Particle accelerators have a significant impact on a variety of scientific fields, including the hunt for novel physics~\cite{gibney2022upgraded}, nuclear waste transmutation~\cite{bowman1998accelerator,nema2011application}, and cancer treatment~\cite{amaldi1999cancer,amaldi2004future}.
Given that they constantly generate large volume of structured data stream throughout operation, particle accelerators are naturally suitable for application of ML techniques which are known to be powerful in handling highly sophisticated input space with precise objectives~\cite{edelen2016neural}.

In recent years, there has been a significant rise in the use of data-driven methodologies in theoretical and technical research around particle accelerators~\cite{arpaia2020machine, edelen2018opportunities}, such as fast and accurate beam dynamics modelling that aims to assist with future accelerator design~\cite{zhao2020beam}, high-precision surrogate models that significantly save computing cost compared to original simulations~\cite{adelmann2019nonintrusive}, beam energy optimization that complies with safely restrictions~\cite{kirschner2019bayesian, kirschner2022tuning}, and more specific use cases including optics correction~\cite{fol2019optics} and collimator alignment automation~\cite{azzopardi2019operational}. 
Among all prospective areas of application, not limited to what is listed above, the assurance of safe and stable operation has always been a crucial concern in accelerator control. Preemptive detection or forecasting of anomalous events during accelerator operation would greatly contribute to more accurate and timely control, longer beam time and better beam quality for the users. Similar approaches have already been thoroughly researched in the field of predictive maintenance~\cite{kang2021remaining}, which aims to identify potential breakdowns in advance and apply the necessary maintenance actions in time. Because of the complexity of the data and wide variations among different accelerators, forecasting the future behavior of an accelerator facility is more difficult to formulate than predicting that of a deteriorating engine. Despite that, there have been multiple attempts across various institutions that address different types of anomalous events, but jointly focus on failure prediction from a preemptive anomaly detection perspective, and open up possibilities for subsequent mitigation measures~\cite{li2023review, edelen2016jgh, edelen2016neural}.

\citet{revsvcivc2020predicting, revsvcivc2022improvements} from the Spallation Neutron Source (SNS) at Oakridge use beam current pulses to detect beam loss trips via binary classification. By taking pulses \emph{before} the errant pulse as the positive class and pulses in no-trip operation as the negative class, prediction of the next pulse's behaviour should be realized inside a time budget of \SI{16}{\milli\second}, which is the interval between two consequent pulses. The best performing Random Forest (RF) classifier, together with Principle Component Analysis (PCA) techniques to refine the input features, reaches 96\% accuracy and 61\% recall, and all classifiers could perform the prediction inside \SI{4}{\milli\second}, much faster than the available time budget. A follow up study from~\citet{blokland2022uncertainty} using the same beam current pulse data achieves uncertainty aware prediction by establishing a Siamese~\cite{koch2015siamese} network. 

\citet{tennant2020superconducting} at the Continuous Electron Beam Accelerator Facility
(CEBAF) at Jefferson Lab have been focusing on  faults in the cryomodules for superconducting radio-frequency (SRF) cavities. Various machine learning models are applied to enable prompt detection of those SRF faults. With a sequential multi-class classification model, the authors further discover that cavity faults which occur quickly are significantly more difficult to anticipate than those that grow gradually~\cite{rahman2022real}. Also dedicated to SRF anomalies, a recent study by \citet{eichler2023anomaly} at the European X-ray Free Electron Laser (EuXFEL) proposes a novel parity space based approach that detects and distinguishes quenches from other anomalies, following the work of ~\citet{nawaz2018anomaly}. A residual signal, obtained from the deviation in RF waveforms and statistically indicated by the generalized likelihood ratio, serves as a distinct measure for various anomaly categories. The method is experimentally implemented, and a detailed analysis  is given on anomalies that are currently mistakenly classified by the existing method.

Following similar methodology, we present in this paper an effective classification-based approach based on a previous study~\cite{li2021novel}, aiming to forecast the beam interruptions, namely \emph{interlocks}, of the proton cyclotron facility --- High Intensity Proton Accelerators (HIPA) at Paul Scherrer Institut (PSI).

\subsection{Facility and Datasets}

HIPA produces a proton beam of nearly \SI{1.4}{\mega\watt} power, which makes it one of the most powerful proton cyclotron facilities in the world~\cite{reggiani2020improving}. To keep track of the accelerator operation while remaining within safety limits, various sensors and monitors are placed across different locations inside the facility, including beam loss monitors, temperature monitors etc. The interlock system is a necessary security measure that immediately shuts off the beam whenever a problem occurs and some signal exceeds their safety limits, such as the loss monitor shown in red in Figure~\ref{fig:interlock}. However, such shut-downs may lead to abrupt operational changes and a substantial loss of beam time. We propose to build a forecasting model for the interlocks as depicted in the bottom right of Figure~\ref{fig:interlock}. Once the model output indicates an incoming interlock, we expect to apply some recovery operation back onto the facility to avoid the interlock from happening, thus save beam time for the users.

\begin{figure}
    \centering
    \includegraphics{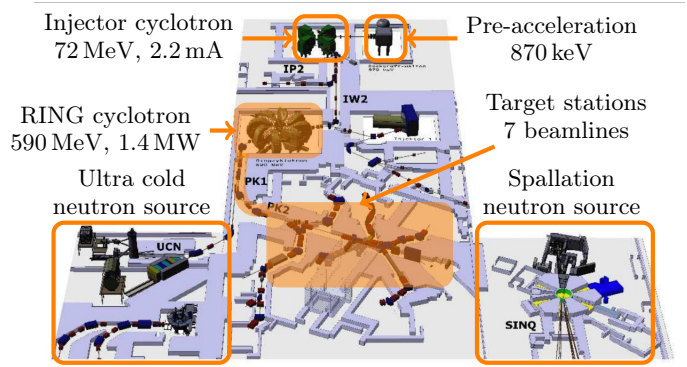}
    \caption{A schematic overview of the HIPA facilities~\cite{kovach2017energy}. The protons are pre-accelerated to \SI{870}{\keV} at a Cockroft-Walton accelerator, then fed into the Injector cyclotron to reach \SI{72}{\MeV}. The beam is then accelerated to the maximum energy of \SI{590}{\MeV} by the large 8-sector RING cyclotron before being transferred to the target stations and experimental regions.}
    \label{fig:hipa}
\end{figure}

\begin{figure}
    \centering
    \includegraphics{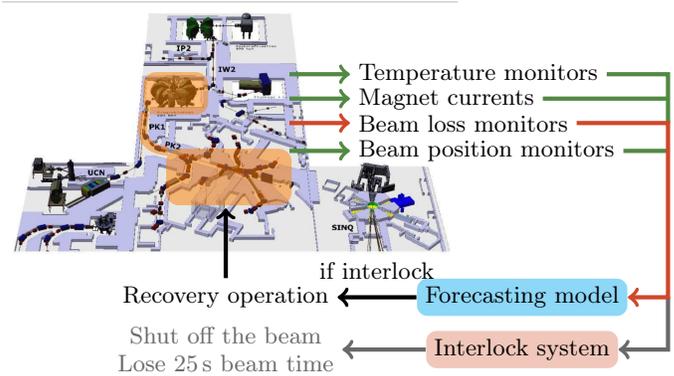}
    \caption{The HIPA interlock system with the proposed forecasting model. In this example, an interlock is triggered when the beam loss monitor sees a beam loss higher than the safety threshold (shown as red). Without the forecasting model, the beam would be immediately shut off and equivalently \SI{25}{\second} beam time would be lost (shown as gray). Whereas if an interlock is predicted by the model, some possible recovery operation could be implemented to circumvent the sudden shutdowns.}
    \label{fig:interlock}
\end{figure}

The data used across this paper are collected from a 53-days period in October and November of 2019 excluding beam development days, during which there were 1192 interlocks in total. The dataset is composed of 376 Process Variables -- so-called \emph{channels} -- and the interlock signals from the Experimental Physics and Industrial Control System (EPICS), which are recorded as multivariate time series and interpolated with \SI{5}{\hertz} frequency. Details of data collection and preprocessing are presented in section 2.1 of~\cite{li2021novel}. We formulate the interlock forecasting problem as a binary classification of two classes of samples, as shown in Fig.~\ref{fig:problem}. The positive class consists of \emph{interlock samples} that are taken close to the interlocks, which represent unstable states. The negative class consists of \emph{stable samples} that are taken far from interlocks and represent stable operating states. 

\subsection{Problem Formulation}
The behaviour of some example channels right before an interlock is shown in Figure~\ref{fig:problem}. The \emph{interlock samples}, shown as orange blocks, are taken at $T_h$ seconds before the interlock, where $T_h$ is thus the forecasting horizon. The \emph{stable samples}, shown as green blocks, are taken during stable operation away from interlocks. Each of these samples is in principle a snippet of the $d=376$ multivariate time series of length $\Delta t$.
\begin{figure}[htb]
\centering
\includegraphics{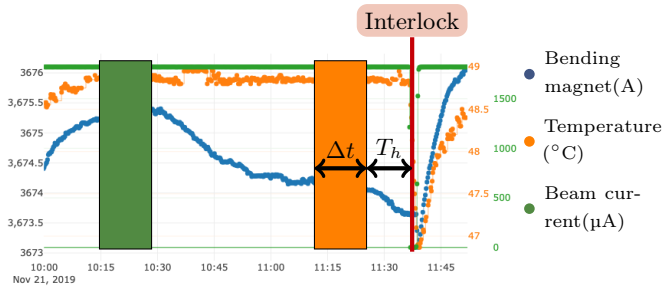}
\caption{Illustration of example channels, interlock and sample taking. 
The orange and green blocks show an \emph{interlock sample} and a \emph{stable sample} respectively. $T_h$ denotes the forecasting horizon, i.e. time of advanced alarm; $\Delta t$ is the duration of each sample.}
\label{fig:problem}
\end{figure}

\subsection{The previous RPCNN model}
In a previous study~\cite{li2021novel}, we have presented the \emph{Recurrence Plot - Convolutional Neural Network} (RPCNN) model that takes window input, transforms time series into images and then performs the classification. However, the problems of a high False Positive rate and unstable performance need to be addressed.

\begin{figure*}[tb]
\centering
\includegraphics[width=\linewidth]{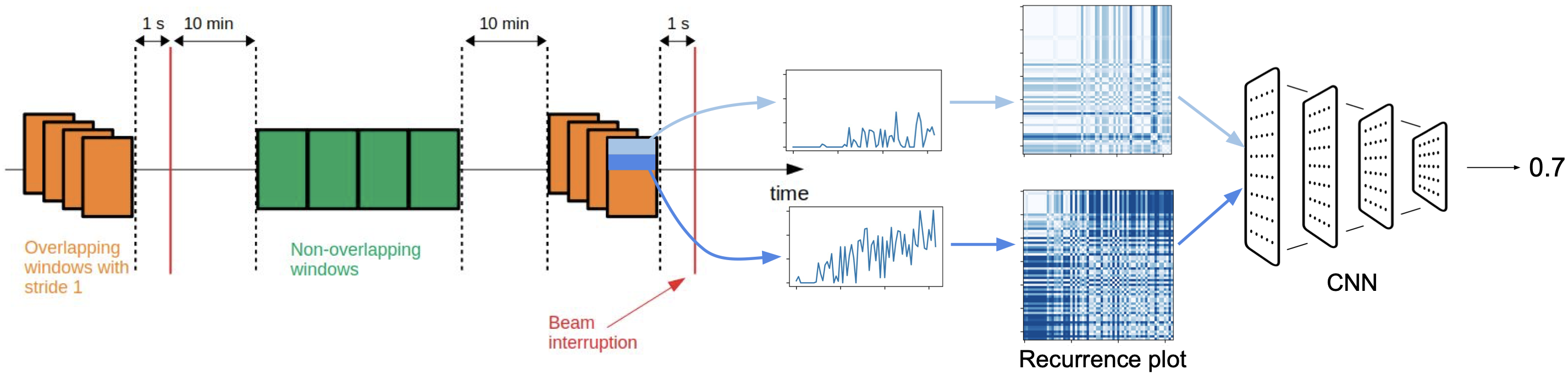}
\caption{The RPCNN model workflow. Two classes of samples are taken as windows from the original 376 channels, with positive samples (orange, indicating interlocks) as sliding windows and negative samples (green, indicating stable operation) as non-overlapping windows. Then each 1-dimensional time series of the windows are transferred into an RP, and in total 376 RPs forms one sample to be fed into the CNN, which gives out a probability score.}
\label{fig:rpcnn}
\end{figure*}

The workflow of the RPCNN model is sketched in Fig.~\ref{fig:rpcnn}. First, we take two classes of time windows either close to or far from the interlocks, where the window length is a tunable model parameter. Explicitly, the \emph{interlock samples} are taken as sliding windows from \SI{1}{\second} to \SI{15}{\second} before interlocks. Second, we transform each 1-dimensional time series of the 376 channels to a 2-dimensional \emph{Recurrence Plot}, which could be interpreted as a pairwise distance measure of the time series and is capable of extracting finer dynamical patterns. Afterwards, we train the plots with \emph{Convolutional Neural Network}, which is proved to be powerful in image-driven pattern recognition~\cite{o2015introduction}. The output is a probability score inside the range $[0,1]$ indicating the likelihood of a sample belonging to the positive (i.e. close to interlock) class. More details such as model architecture are published in~\cite{li2021novel}.

Typical binary classification metrics in a confusion matrix are defined and applied in our setting. A True Positive (TP) means an interlock sample --- a sample less than 15~s before an interlock --- being classified as an interlock. A False Positive (FP) means a stable sample --- a sample at least 10 minutes away from an interlock --- being mistaken as an interlock. A True Negative (TN) means a stable sample being reported as stable. And the remaining False Negative (FN) means that the model fails to identify an interlock.

For the predictions of a binary classification model, the ROC curve shows its true positive rate ($TPR=\frac{TP}{TP+FN}$) against false positive rate ($FPR=\frac{FP}{FP+TN}$) as a function of varying classification threshold~\cite{fawcett2006introduction}. The Area Under Curve (AUC) is a score satisfies $AUC \in [0,1]$ which represents the probability that a random positive (i.e. interlock) sample receives a higher value than a random negative (i.e. stable) sample. It shows the general ability of a model to distinguish between the two classes when taken all classification thresholds into account. The model parameters are chosen according to highest mean AUC value (AUC=0.71) of 25 model instances with random initialization. The classification threshold --- which directly determines the confusion matrix --- is generally chosen at a point on the ROC curve that leverages between high true positive rate and low false positive rate. 

Since the long term goal is to integrate the model into real-time operation, a way to quantify the model performance in real-time setting is needed in addition to standard classification metrics. We develop a custom metric called \emph{average Beam Time Saved per interlock} $\overline{T_s}$, taking real beam time loss and potential recovering methods into consideration. It is assumed that an interlock could be circumvented by reducing 10\% of the beam current, based on experiments and expert consultation. The classification threshold is then chosen according to largest $\overline{T_s}$. However, this metric imposes strict constraint on false positives. FPR is brought down to below 0.2\% at a cost of true positive rate reaching only 4.9\%. As shown later in Fig.~\ref{fig:roc} and Table~\ref{tab:realtime}, the classification result is neither as expected nor stable. There are still a large number of unwanted false positives; and over-fitting might also have occurred due to the non-negligible uncertainty in results as well as unsatisfactory convergence in hyperparameter scan. We examine the input data more in-depth with visualization, expert knowledge and statistical tests to find out the reason, and the customized real-time metric is also revised.

\section{Maximum Mean Discrepancy (MMD) and Two Sample Test}

To discover how long before the interlocks the precursors appear, we take the input data at $t_0$ and $t_1$ seconds respectively before all interlocks, and group them together into two sets. The problem is then formulated as a statistical problem --- to discover whether two random variables are sampled from an identical distribution, without imposing any assumption on the unknown true distribution. To achieve such a goal, we perform \emph{two sample test}~\cite{gretton2012kernel} on the two sets taken at different time before the interlocks, and statistically compare their \emph{Maximum Mean Discrepancy} (MMD). MMD could be interpreted as the distance between two distributions, and a large MMD implies a large difference between the two sets. Figure~\ref{fig:test_outline} generally describes the process.

\begin{figure}
\centering
\includegraphics[width=\linewidth]{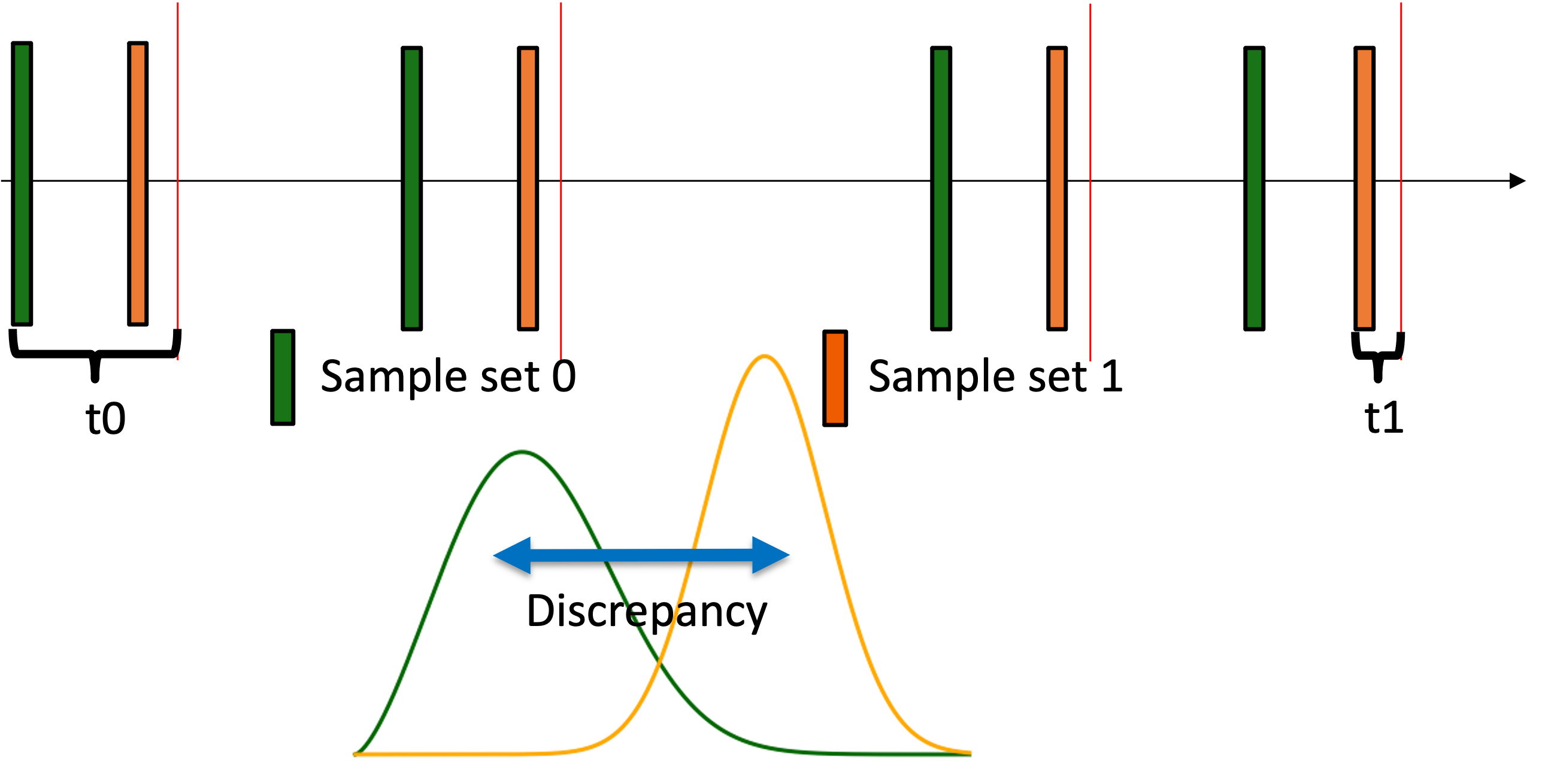}
\caption{An outline of the two sample MMD test that shows the procedure of taking two dataset from $t_0$ and $t_1$ seconds before interlocks.}
\label{fig:test_outline}
\end{figure}

Two sample tests are a group statistical tests that aim to decide whether two sets are drawn from the same distribution. Namely, consider samples $X=\{x_1,\cdots, x_m\}$ that are independently and identically drawn (i.i.d.) from distribution $p$, and samples $Y=\{y_1,\cdots,y_n\}$ i.i.d. from distribution $q$, the null hypothesis of two sample test is $H_0: p = q$, and the alternative hypothesis is thus $H_1: p \neq q$. Examples of widely applied two sample tests include the two sample Kolmogorov-Smirnov (K-S) test~\cite{massey1951kolmogorov}, which compares the empirical cumulative distribution functions between two samples; Anderson-Darling (A-D) test~\cite{scholz1987k}, which is based on K-S test and gives more weights on the tail of distribution; and Cram\'{e}r-von Mises test~\cite{darling1957kolmogorov}, which employs a slightly different measure than A-D test. The tests mentioned above are only suitable for univariate problems; otherwise a multivariate generalisation step is necessary, as \citet{friedman1979multivariate} have shown using Minimum Spanning Trees or, more broadly, Nearest Neighbors approaches.

We choose to apply the two sample test based on MMD~\cite{gretton2012kernel} in our high dimensional data. Its application has been expanding recently, together with the growing trend of data complexity and development of machine learning, for instance the identification of lung cancer as in~\cite{zhao2019identification} and in the training of adversarial neural networks~\cite{dziugaite2015training}.
MMD is a distance metric for the closeness of two probability distributions. Specifically, it is the largest difference between expectations of functions in the unit ball of a reproducing kernel Hilbert space (RKHS)~\cite{sriperumbudur2010hilbert, sriperumbudur2011universality}, which is a Hilbert space of functions that fulfils certain continuity conditions.

Formally, we have observations $X_0$ and $X_1$ from a topological space $\mathcal{X}$, while $p$ and $q$ are their Borel probability measures, respectively. (In our practical case, $\mathcal{X}=\mathbb{R}^d$ where $d=376$ is the input dimension or number of channels.) The maximum mean discrepancy (MMD) is defined as
\begin{equation}\label{eq:mmd1}
    \text{MMD}(\mathcal{F}, p, q) := \sup_{f\in \mathcal{F}}\left(\mathbf{E}_{x_0}[f(x_0)]-\mathbf{E}_{x_1}[f(x_1)]\right)
\end{equation}
where $\mathcal{F}$ is a class of functions $f: \mathcal{X} \to \mathbb{R}$, and $\mathbf{E}_{x_0}[f(x_0)]$ is simplified from $\mathbf{E}_{{x_0}\sim p}[f(x_0)]$, denoting the expectation taking over the distribution $p$. To serve as a test metric capable of distinguishing between distributions, the function class $\mathcal{F}$ needs to be comprehensive enough to ensure a unique value when $p=q$ while being constrained enough to allow for practical finite estimates. \citet{gretton2012kernel} propose $\mathcal{F}$ to be the unit ball in an RKHS~$\mathcal{H}$. Based on the properties of $\mathcal{H}$, the \emph{mean embedding} $\mu_p \in \mathcal{H}$ of a probability distribution $p$ is introduced as
\begin{equation}
    \mathbf{E}_{x_0} f = \langle f, \mu_p \rangle_{\mathcal{H}}
\end{equation}
where $\langle\cdot,\cdot\rangle_{\mathcal{H}}$ denotes the inner product in $\mathcal{H}$. It can then be shown that MMD equals the distance between the two mean embeddings
\begin{equation}\label{eq:mmd2}
    \text{MMD}(\mathcal{F}, p, q) = \|\mu_p - \mu_q\|_{\mathcal{H}}.
\end{equation}

It can be further proved that if $\mathcal{F}$ is the unit ball in a universal RKHS~$\mathcal{H}$ on the compact metric space $\mathcal{X}$ with associated continuous kernel $k(\cdot, \cdot)$, then MMD$(\mathcal{F}, p,q)=0$ iff $p=q$ and the uniqueness is guaranteed~\cite{gretton2012kernel, sriperumbudur2010hilbert}. A detailed derivation of MMD is given in the appendix.

The mean embedding $\mu_p$ is an infinite-dimensional representation of the distribution $p$. Calculating MMD as the RKHS norm of the difference between the mean embeddings contains an essential Fourier transform step, which would require evaluating at infinite frequencies. Heuristic pseudo-distances seem like a plausible solution, such as the difference between values of characteristic functions at a single frequency~\cite{heathcote1972test}. However such methods are highly dependent on prior knowledge of the distributions to compare; and local agreement at some frequency intervals might cover up the globally distinct characteristic functions, thus causing the consistency issue.

Therefore in practice, the actual statistical test we adapt and implement is a variation called the \emph{mean embeddings} (ME) test, based on differences between analytic mean embeddings~\cite{chwialkowski2015fast}. Instead of integrating over the whole frequency domain, the ME test guarantees that it is \emph{almost surely} sufficient to evaluate only at a single random frequency. Theoretically, it is proved to be \emph{almost surely} consistent for all distributions; on the computational side, it has only linear time complexity with respect to the sample size, compared to quadratic time for the original MMD test~\cite{gretton2012kernel}. 

The process of the ME test is described as follows
\begin{enumerate}
    \item Randomly take a set of $J$ test locations $\{\mathbf{v}_j\}_{j=1}^{J} \subset \mathcal{X}$ (i.e. $\mathbb{R}^d$) according to certain optimization schemes~\cite{chwialkowski2015fast};
    \item Take a positive definite kernel $k: \mathcal{X}\times \mathcal{X} \to \mathbb{R}$ (normally use the Gaussian kernel);
    \item For each pair of original samples from the two sets $\{x_{0,i}, x_{1,i}\}_{i=1}^n$, calculate the \emph{differences between kernels}
    \begin{align}
        \vec{z}_i &:= \bigg(k(x_{0,i}, v_1)-k(x_{1,i}, v_1), \dots, \nonumber\\
        &k(x_{0,i}, v_J)-k(x_{1,i}, v_J)\bigg) \in \mathbb{R}^J
    \end{align}
    which is a $J-$dimensional vector.
    \item Calculate the \emph{mean empirical differences}
    \begin{equation}
        \vec{w}_n = \frac{1}{n}\sum_i  \vec{z}_i
    \end{equation}
    and its \emph{covariance matrix}
    \begin{equation}
        \mathbf{\Sigma}_n = \frac{1}{n-1}\sum_i(\vec{z}_i-\vec{w}_n)(\vec{z}_i-\vec{w}_n)^T.
    \end{equation}
    \item The test statistics is computed as
    \begin{equation}
        \widehat{\lambda}_n := n\vec{w}_n\mathbf{\Sigma}_n^{-1}\vec{w}_n
    \end{equation}
    which follows a $\chi^2$ distribution with $J$ degrees of freedom under the null hypothesis $H_0$;
    \item Choose a threshold $t_\alpha$ corresponding to the $1-\alpha$ quantile of the $\chi^2(J)$ distribution, and reject $H_0$ whenever $\widehat{\lambda}_n$ is larger than $t_\alpha$.
\end{enumerate}

To examine when the difference in the input data is emerging, we take samples $X_t=\{x_t^{I_0}, x_t^{I_1}, x_t^{I_2}, \cdots\}$ at $t$ seconds before all the interlocks. Each sample is a $d=376$ dimensional vector, where $t = \SI{0.2}{\second}, \SI{0.4}{\second}, \dots$ denotes the time before interlocks when we take the sample, and $I_0, I_1, \dots$ denotes the interlock number. The sample size for each $X_t$ is the number of available interlocks in the training data. We perform two sample ME test on each pair of samples, namely calculate the ME test statistics between $X_{t=\SI{0.2}{\second}}$ with $X_{t=\SI{0.4}{\second}}, X_{t=\SI{0.6}{\second}}, \dots$, then perform the test between $X_{t=\SI{0.4}{\second}}$ with $X_{t=\SI{0.6}{\second}}, X_{t=\SI{0.8}{\second}}, \dots$ and so forth. Table~\ref{tab:mmd} lists values of the test statistics with $J=5$ and $\alpha=0.01$ for the two sets taken at different $t_0$ and $t_1$, and Figure~\ref{fig:mmd2d} shows the full matrix of test statistics values of all pairs of $t_0$ and $t_1$ from \SI{0.2}{\second} to \SI{10}{\second} on log scale. It is clearly visible that samples taken at \SI{0.2}{\second} before interlocks are notably different, while \SI{0.4}{\second} also shows some slight difference.

\begin{table}
\centering
\caption{Value of $\widehat{\lambda}_n$ ($J=5$, $\alpha=0.01$) for various pair of two sets taken at $t_0$ and $t_1$ seconds before interlocks. The bold numbers denote that $H_0$ is rejected, i.e. the two sets are considered different.}
\begin{tabular}{lrrrrrr}
\toprule
$\mathbf{t_1\big\backslash t_0}$ & \SI{0.2}{\second} & \SI{0.4}{\second} & \SI{0.6}{\second} &  \SI{0.8}{\second} & $\dots$ &\SI{10.0}{\second}\\
\midrule
\SI{0.2}{\second} & 1.2 & \textbf{4935.6} & \textbf{8102.1} & \textbf{7999.4} & $\dots$ & \textbf{7330.0}\\
\SI{0.4}{\second} & - & 5.6 & \textbf{177.0} & \textbf{165.3} & $\dots$ & \textbf{215.8}\\
\SI{0.6}{\second} & - & - & 2.6 & 3.7 & $\dots$ & \textbf{103.9} \\
\SI{0.8}{\second} & - & - & - & 3.6 & $\dots$ & \textbf{93.7}\\
\bottomrule
\end{tabular}
\label{tab:mmd}
\end{table}

\begin{figure}
    \centering
    \includegraphics[width=\linewidth]{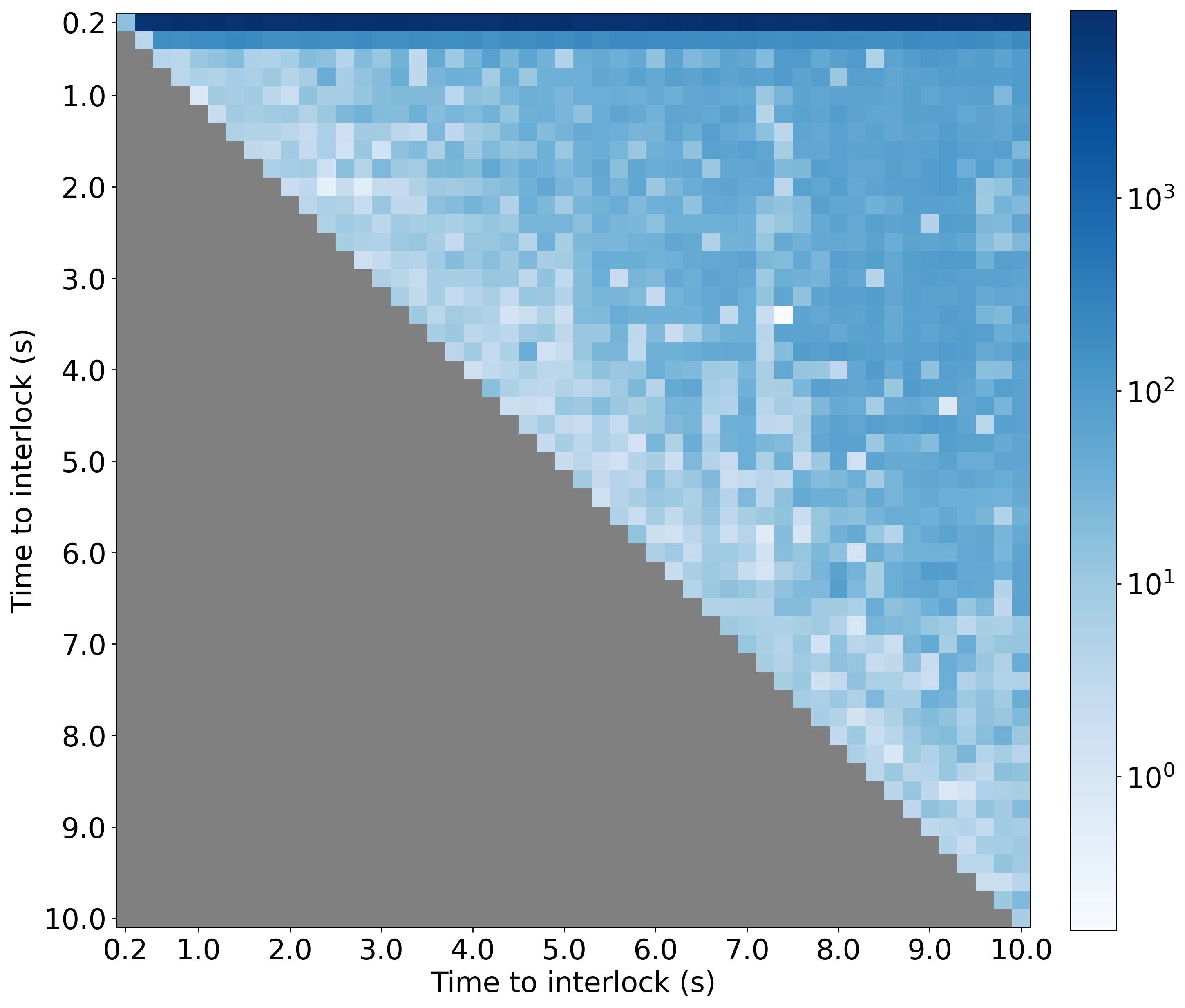}
    \caption{Pairwise ME test statistics comparing sample sets taken at $t_0$ and $t_1$ seconds before interlocks, ranging from \SI{0.2}{\second} to \SI{10}{\second}. The values are transformed on log scale to emphasize the difference. A darker blue color implies more significant discrepancy between the two compared distributions. The lower triangle is masked as gray due to symmetry.}
    \label{fig:mmd2d}
\end{figure}

\begin{figure}
    \centering
    \includegraphics[width=\linewidth]{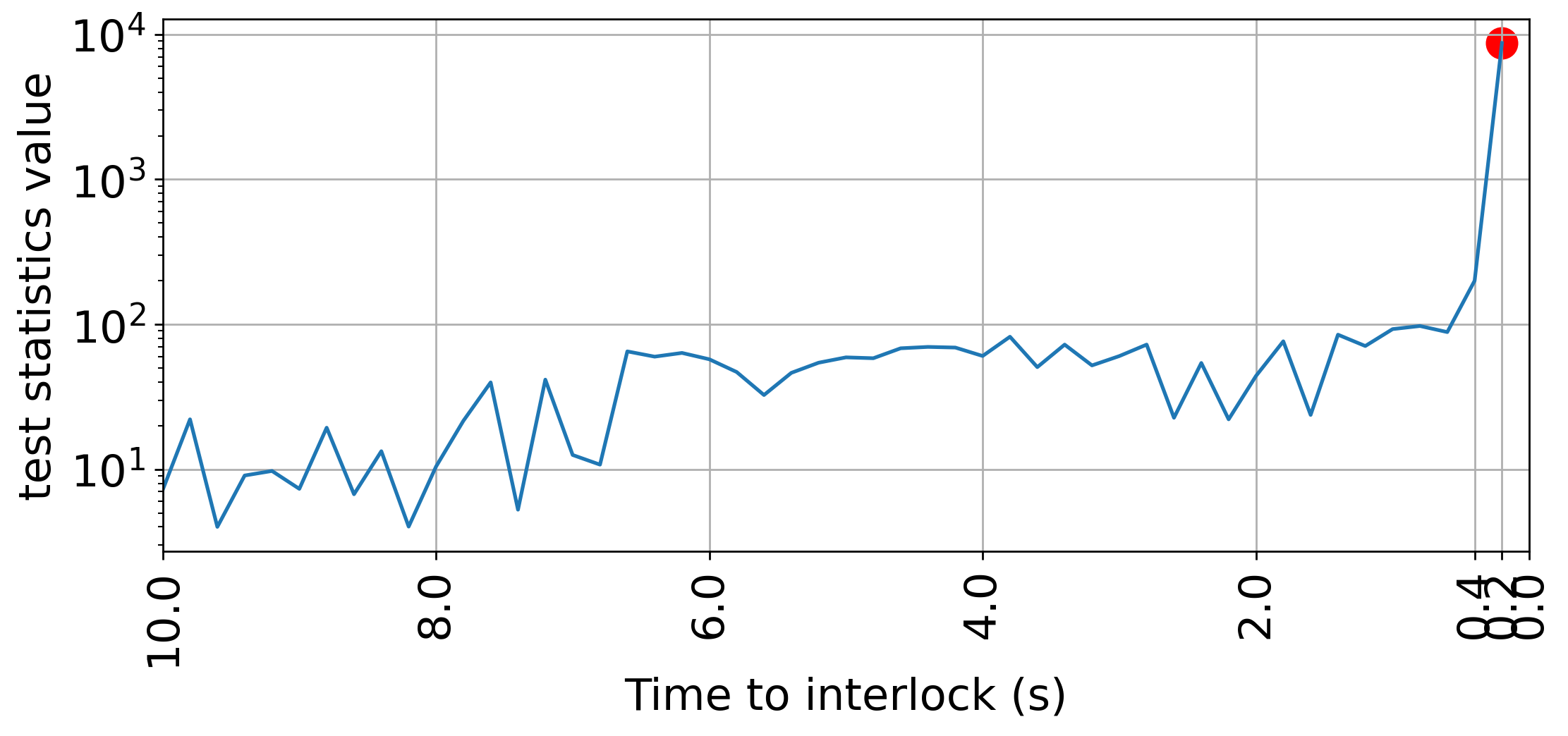}
    \caption{The test statistics value of the set $t_0$ fixed at \SI{10}{\second} compared with data taken at \SI{9.8}{\second} towards \SI{0.2}{\second} before all interlocks in 53 days of 2019. With a log scale for the test statistics, we see that changes are only apparent about \SI{0.2}{\second} and \SI{0.4}{\second} before the interlocks.
}
    \label{fig:mmd10}
\end{figure}

Figure~\ref{fig:mmd10} shows the ME test statistics between samples taken at $t_0=\SI{10}{\second}$ before interlocks and samples taken from $t_1=\SI{9.8}{\second}$ towards \SI{0.2}{\second} before interlocks. With the log scale in y-axis, it is clearly seen that noticeable changes are only present less than \SI{0.4}{\second}, or even only \SI{0.2}{\second} before interlocks. The result implies that interlocks are mostly \emph{abrupt} events and the time scale of changes are much less than \SI{1}{\second}. This explains the weak classification power of RPCNN, since its input data are taken earlier than \SI{1}{\second} before interlocks, which fails to capture the abrupt changes. Since the data remain quite similar most of the time, the model actually sees similar inputs for the two different classes, and its large number of model parameters may only capture spurious patterns specific to the training set. This might as well explain the issue of over-fitting in RPCNN.

\section{The LASSO model}
Based on the above MMD analysis, we modify the two classes of our previous binary classification model. The positive and negative classes are composed of samples taken at $t_1$ and $t_0$ seconds before interlocks, respectively. Instead of the previous time windows, each sample now becomes a $(376,)$ vector taken at one single timestamp. We fit these two classes of samples with a LASSO regression model, which is a standard logistic regression plus a $L_1$ regularization term to suppress the number of input features. 

The inputs are denoted as $\{x_i \in \mathbb{R}^d \}_{i=1}^n $ where $d=376$ and $n$ is the number of samples. The class labels are $\{y_i\}_{i=1}^n \in \{\pm 1\}$, where $y=1$ for positive samples (i.e. $t_1$ seconds before interlock) and $y=-1$ for negative samples (i.e. $t_0$ seconds before interlock). The goal is to fit weight $\mathbf{\omega} \in \mathbb{R}^d$ by minimizing the loss in Eq.~\eqref{eq:lasso}
\begin{align}\label{eq:lasso}
    \min_{\mathbf{\omega}}{L} &= \min_{\mathbf{\omega}}\big[\underbrace{\frac{1}{n}\sum_{i=1}^n\log{[1+\exp{(-y_i \mathbf{\omega}^T \cdot x_i)}]}}_{\text{logistic loss}} \nonumber \\
    \quad &+\underbrace{\lambda \| \mathbf{\omega}\|}_{\text{regularization}}\big].
\end{align}

The output is again a probability output inside the range $[0,1]$. Compared to RPCNN, LASSO is simple, linear and sparse, and the regularization could also lead to better model interpretability.

We compare different choices of $t_0$ and $t_1$ from $\{\SI{0.2}{\second}, \SI{0.4}{\second}, \SI{0.6}{\second}, \SI{1.0}{\second}, \SI{10.0}{\second}\}$ before all interlocks, then train a LASSO logistic regression model to do the classification with 5-fold cross validation, together with a grid-search on the choice of regularization parameter $\lambda$. The summary in Table~\ref{tab:confusion} combines the classification results from all of the 5-folds when they are hold-out as test set. The subscription $c$ for the classification metrics TPR and FPR in the table refers to typical \emph{classification} metric calculated within input samples, in order to distinguish from the customised \emph{real-time} metric introduced later in Section~\ref{sec:realtime}.

\begin{table}
\centering
\caption{Binary classification results shown as TPR\textsuperscript{c} (\%) | FPR\textsuperscript{c} (\%), comparing samples taken at various $t_0$ and $t_1$ seconds before interlocks. The dataset is taken from October and November of 2019 with 1192 interlocks in total. The best TPR\textsuperscript{c} and FPR\textsuperscript{c} results are marked in bold.}
\begin{tabular}{lllll}
\toprule
\multicolumn{1}{c}{$\mathbf{t_1\big\backslash t_0}$}  & \multicolumn{1}{c}{\SI{0.4}{\second}} & \multicolumn{1}{c}{\SI{0.6}{\second}} &  \multicolumn{1}{c}{\SI{1.0}{\second}} &\multicolumn{1}{c}{\SI{10.0}{\second}}\\
\midrule
\SI{0.2}{\second} & 93.2 | 2.3 & 96.8 | 0.5 & 96.8 | \textbf{0.3} & \textbf{97.4} | 0.6\\
\SI{0.4}{\second} & \multicolumn{1}{c}{-} & 52.8 | 22.5 & 53.2 | 21.6 & 57.4 | 23.1\\
\SI{0.6}{\second} & \multicolumn{1}{c}{-} & \multicolumn{1}{c}{-} & 56.0 | 44.2  & 58.4 | 41.6 \\
\SI{1.0}{\second} & \multicolumn{1}{c}{-} & \multicolumn{1}{c}{-} & \multicolumn{1}{c}{-}  & 57.8 | 43.2\\
\bottomrule
\end{tabular}
\label{tab:confusion}
\end{table}

We choose the $t_1=\SI{0.2}{\second}$ versus $t_0=\SI{10}{\second}$ model to proceed for real-time implementation since it has the highest true positive rate. The positive class is now composed of samples taken only at \SI{0.2}{\second} before interlocks, while the negative class is taken at \SI{10}{\second} before interlocks. 
\section{MODEL EVALUATION}
We present two types of evaluation metrics: the Receiver Operating Characteristic (ROC) curve which is standard for binary classification, and our custom evaluation metric \emph{Beam Time Saved} defined in real-time scenario.

\subsection{Classification metric}
Following~\citet{li2021novel}, we keep evaluating the performance of the binary classification model by the ROC curve and its corresponding AUC value. For both RPCNN and Lasso,  Figure~\ref{fig:roc} shows the ROC curve and AUC together with uncertainties calculated from 25 trials each. The behaviour of a perfect model and random guess is also plotted for comparison. With an AUC of $0.99$ and much narrower uncertainty, LASSO demonstrates its superiority in terms of better classification power as well as more stable performance compared to RPCNN.
\begin{figure}[!tbh]
    \centering
    \includegraphics[width=\linewidth]{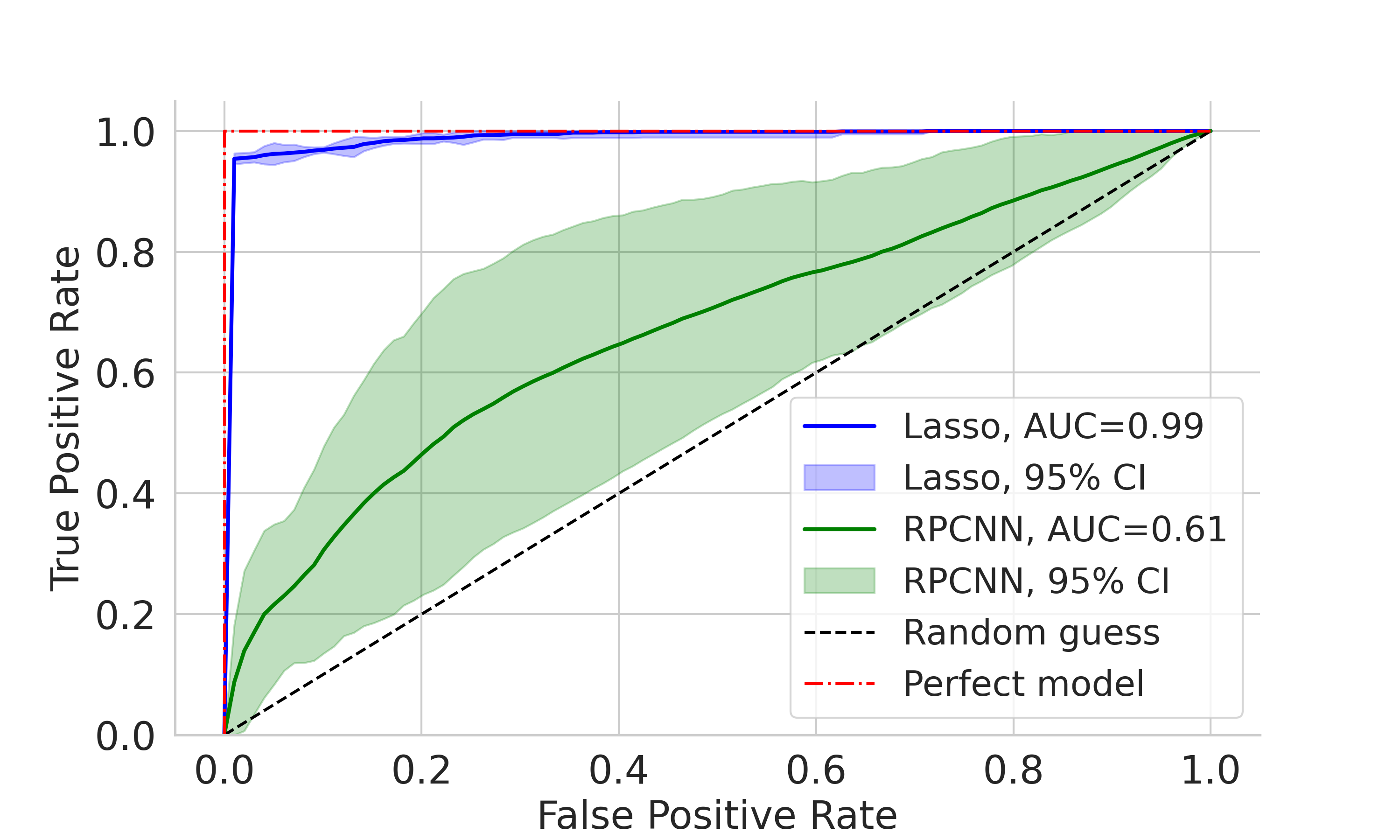}
    \caption{The ROC curves and AUC for LASSO and RPCNN models, a perfect classifier and random guess. The 95\% confidence interval is also depicted as areas.}
    \label{fig:roc}
\end{figure}

A point worth noting is the difference of the inputs of RPCNN from the Lasso model. RPCNN takes window input at least \SI{1}{\second} before interlocks; however, according to the findings from the MMD test that changes only take place inside \SI{0.2}{\second}, the interlock windows and stable windows --- which are considered as two difference classes for RPCNN --- are similarly distributed. This explains the big instabilities and large amount of false positives in prediction. 

\subsection{Real-time metric}\label{sec:realtime}
In the ROC curve, the number of TP\textsuperscript{c} and FP\textsuperscript{c} are calculated only from the testing dataset. In real operation, these metrics need to be updated as new data and interlock events are recorded continuously, and it should be possible to evaluate the model according to an adjustable forecasting horizon. To adapt and improve the previous \emph{beam time saved} metric, we propose our customised definition of TP and FP in real time (denoted as TP\textsuperscript{r} and FP\textsuperscript{r}) as shown in Figure~\ref{fig:goal}.

\begin{figure}
    \centering
    \includegraphics[width=\linewidth]{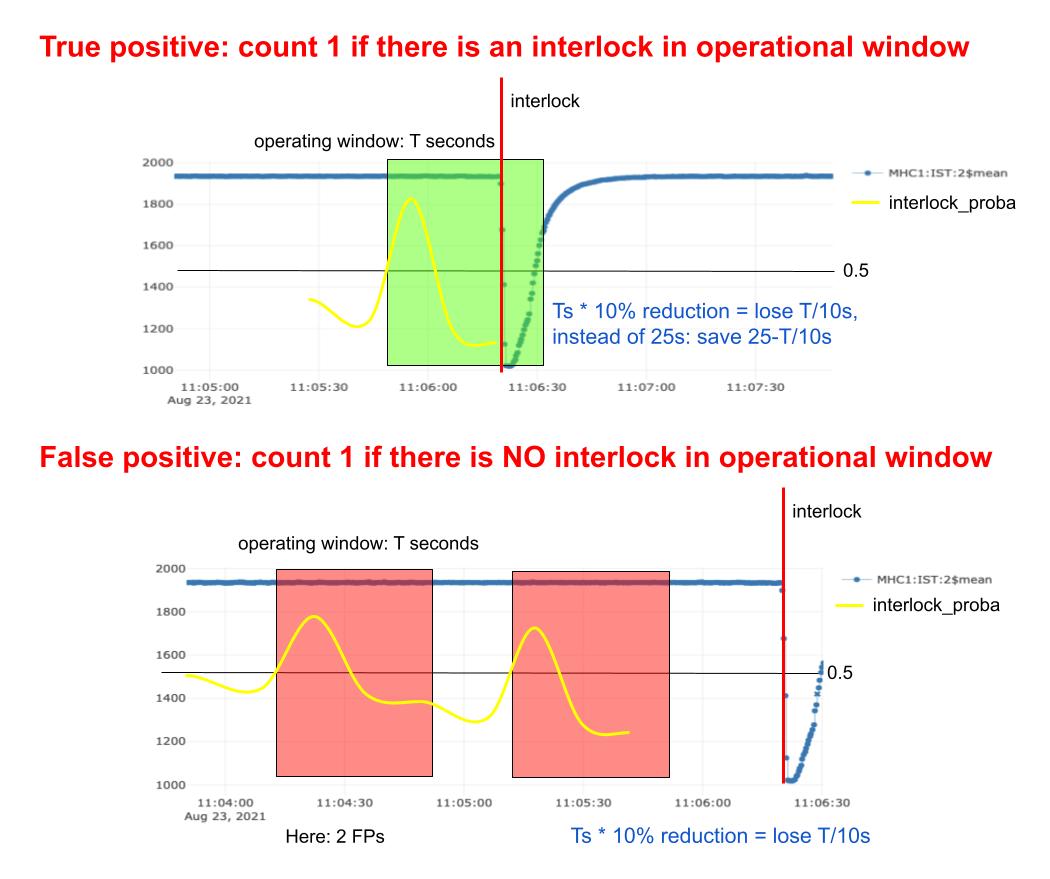}
    \caption{Customised definition of real-time true positives (TP\textsuperscript{r}) and false positives (FP\textsuperscript{r}). An inspection window of \SI{1}{\minute} starts when the model output goes above threshold. A TP\textsuperscript{r} is counted if interlocks fall in this inspection window, otherwise a FP\textsuperscript{r} is recorded.}
    \label{fig:goal}
\end{figure}

Assume the model output is above the classification threshold at time $t$. If there is actually an interlock inside an inspection window of \SI{1}{\minute} after $t$, then this interlock is successfully predicted and it is counted as one TP\textsuperscript{r}; otherwise, if there is no interlock for \SI{1}{\minute}, then one FP\textsuperscript{r} is counted. The number of TP\textsuperscript{r} cannot exceed the total number of interlocks, while the number of FP\textsuperscript{r} is not limited. Thus the ratio TPR\textsuperscript{r}$=N_{\text{TP}_r}/N_{int}$ could be used to evaluate model performance, where $N_{int}$ denotes the total number of interlocks. Figure~\ref{fig:tpfp} shows examples of TP\textsuperscript{r} and FP\textsuperscript{r} in real-time model operation, complying with the procedure described in Figure~\ref{fig:goal}.

\begin{figure*}
    \centering
    \includegraphics[width=\linewidth]{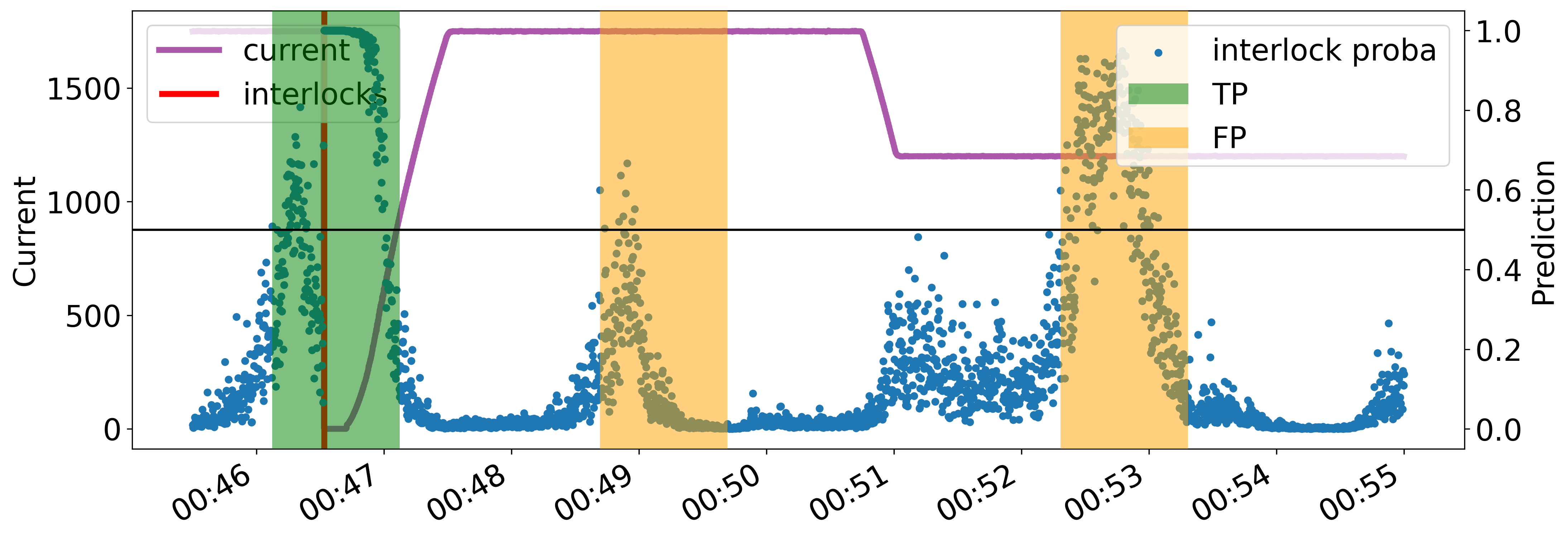}
    \caption{Examples of real-time TP\textsuperscript{r} and FP\textsuperscript{r} from one LASSO training that compares two sets of $t_1=\SI{0.2}{\second}$ and $t_0=\SI{0.4}{\second}$ with 0.5 as the classification threshold, taken from \SI{10}{\minute} period on 2019-10-06. Note that this model is not the best performing one according to the classification and real-time metrics, only to manifest the definition of TP\textsuperscript{r} and FP\textsuperscript{r}. In this example, the interlock is successfully predicted around \SI{30}{\second} in advance, but there are also 2 FP\textsuperscript{r} generated by the model. The \emph{beam time saved} $T_s^r$ during this period is thus $19\cdot 1-6\cdot 2=7$ seconds. }
    \label{fig:tpfp}
\end{figure*}

With the customized definition of real-time TP\textsuperscript{r} and FP\textsuperscript{r}, we construct the real-time metric \emph{Beam Time Saved} $T_s^r$ in Eq.~\eqref{eq:ts} during any given time period
\begin{equation}\label{eq:ts}
    T_s^r \vcentcolon= 19 \cdot \text{TP}^{r} - 6 \cdot \text{FP}^{r}
\end{equation}
where TP\textsuperscript{r} and FP\textsuperscript{r} are the number of real-time TPs and FPs during the concerned time period. The numbers $19$ and $6$ come from the fact that a normal interlock causes a beam time loss of about \SI{25}{\second}; and the assumption that an interlock can be circumvented by reducing the beam current by 10\% is equivalent to a beam time loss of \SI{6}{\second}, as illustrated in Figure~\ref{fig:metric}. One TP\textsuperscript{r} therefore saves \SI{19}{\second} beam time, yet one FP\textsuperscript{r} loses \SI{6}{\second} beam time, as shown in Table~\ref{tab:metric}. 

\begin{figure}
\centering
\includegraphics[width=\linewidth]{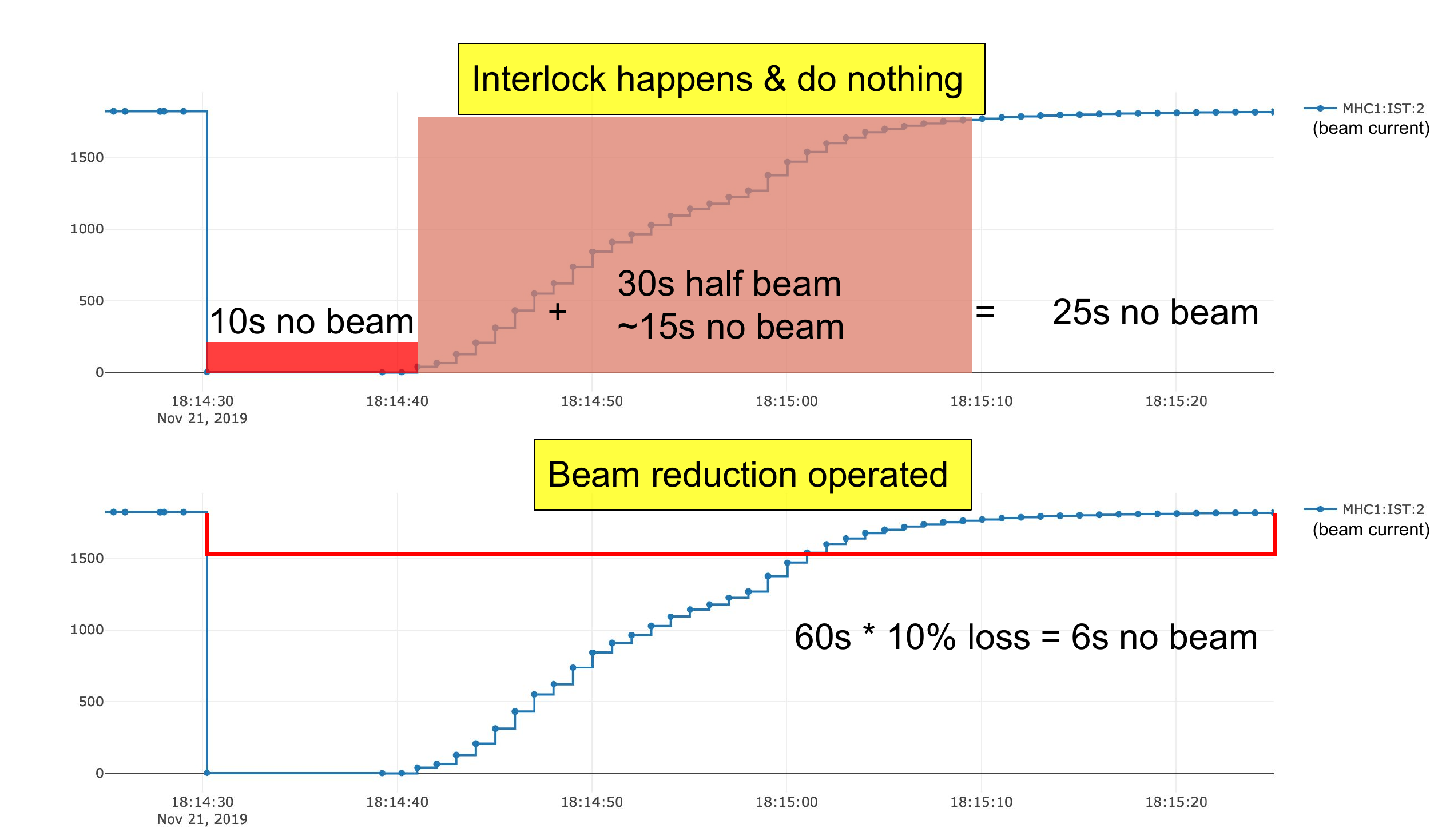}
\caption{Lost beam time in case of interlocks and when beam reduction is operated. Since an interlock lasts around \SI{1}{\minute}, reducing 10\% of beam current is equivalent to \SI{6}{\second} beam time loss.}
\label{fig:metric}
\end{figure}

\begin{table}[!hbt]
   \centering
   \caption{Summary of lost and saved beam time (\SI{}{\second}) in case of TP\textsuperscript{r}, FP\textsuperscript{r} and FN\textsuperscript{r}, before and after the 10\% beam reduction is performed.}
   \begin{tabular}{lrrr}
       \toprule
       \textbf{Conditions} & \textbf{TP\textsuperscript{r}} & \textbf{FP\textsuperscript{r}} &  \textbf{FN\textsuperscript{r}} \\
       \midrule
No recovery operation & -25 & 0 & -25\\
10\% beam current reduced & -6 & -6 & -25 \\
\cdashlinelr{1-4}
Net beam time saved (\SI{}{\second})& 19 & -6 & 0 \\
\bottomrule
   \end{tabular}
   \label{tab:metric}
\end{table}

Table~\ref{tab:realtime} lists the above-mentioned real-time metrics of the two models. LASSO performs better in all metrics, and it has the potential to save around \SI{5}{\minute} of beam time per day, depending on the chosen classification threshold. The numbers of TP\textsuperscript{r} (thus TPR\textsuperscript{r} as well) and FP\textsuperscript{r} drop with increasing threshold, as this reflects a stricter criteria for positives. The trade-off between TP\textsuperscript{r} and FP\textsuperscript{r} in $T_s^r$, as expressed by Eq. \eqref{eq:ts}, portends the existence of an optimal choice of the classification threshold --- the threshold of 0.8 is chosen from current experiments to reach the largest $T_s^r=\SI{5.45}{\minute/\day}$. The best-performing RPCNN model~\cite{li2021novel} has a negative gain of \SI{0.27}{\minute/\day} more beam time loss, due to its rather low number of TP\textsuperscript{r} as well as high number of FP\textsuperscript{r}. Also the different definitions of TP, FP and $T_s$ play a role here.

\begin{table}[!hbt]
   \centering
   \caption{Real-time metrics of both models. The best-performing RPCNN model~\cite{li2021novel}, which is selected based on the previous \emph{average beam time saved} definition $\overline{T_s}$, is listed here as a benchmark. The subscripts $0.5$ to $0.9$ under LASSO denote the different choices of classification thresholds. The best values among each column are marked bold.}
   \begin{tabular}{lrrrrr}
       \toprule
       \textbf{Model} & \textbf{TP\textsuperscript{r}} & \textbf{ FN\textsuperscript{r}} &
       \textbf{TPR\textsuperscript{r}}(\%) &  \textbf{FP\textsuperscript{r}} & $\mathbf{T_s^r}$ (\SI{}{\minute/\day}) \\
       \midrule
RPCNN\textsubscript{best} & 99 & 1093  & 8.3 & 455 & -0.27\\
LASSO\textsubscript{0.5} & \textbf{1017} & \textbf{175} & 
\textbf{85.3} & 1405 & 3.43 \\
LASSO\textsubscript{0.6} & 1004 & 188 &
84.2 & 784 & 4.52 \\
LASSO\textsubscript{0.7} & 1000 & 192 &
83.9 & 429 & 5.17 \\
LASSO\textsubscript{0.8} & 983 & 209 &
82.5 & 222 & \textbf{5.45} \\
LASSO\textsubscript{0.9} & 956 & 236 &
80.2 & \textbf{182} & 5.37 \\
       \bottomrule
   \end{tabular}
   \label{tab:realtime}
\end{table}

Figure~\ref{fig:perf_realtime} compares the performance of the RPCNN\textsubscript{best} and LASSO\textsubscript{0.8} models in real-time scenario. The data are taken from 3 to 4 A.M. on October $6^{th}$ 2019, with 2 interlocks in total. LASSO succeeds in capturing both of the two TPs, while RPCNN misses one TP, and more FPs are clearly present for RPCNN as well.

\begin{figure*}
\centering
\includegraphics[width=\linewidth]{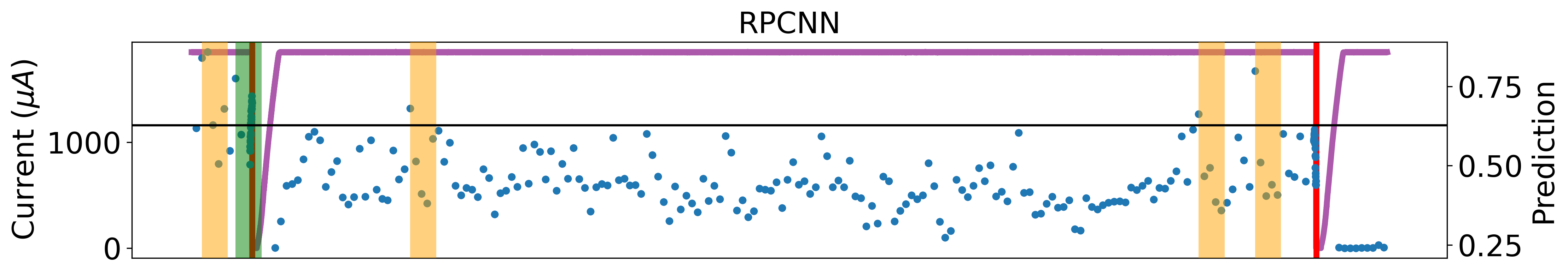}
\includegraphics[width=\linewidth]{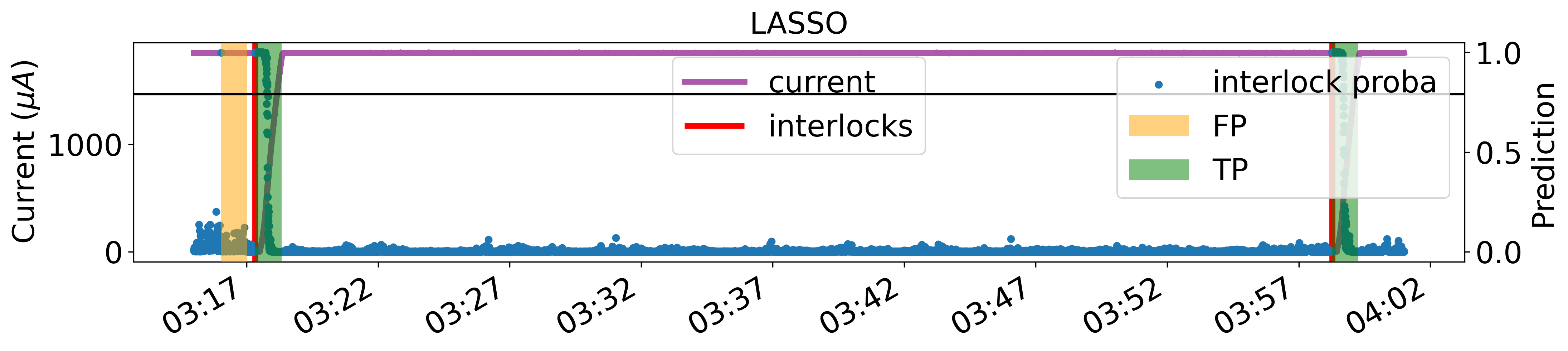}
\caption{Real-time performance of both models in one hour in 2019-10-06.}
\label{fig:perf_realtime}
\end{figure*}

Since the positive samples are only taken at $t_1=\SI{0.2}{\second}$, most of interlocks are expected to be predicted only \SI{0.2}{\second} in advance. However, still some earlier precursors appear in the real-time evaluation of the model, and the actual forecasting horizon $T_h$ (as introduced in Figure~\ref{fig:problem}) extends beyond \SI{0.2}{\second}.  Figure~\ref{fig:horizon} is a cumulative distribution plot of the number of successfully forecasted interlocks (i.e. the 983 TP\textsuperscript{r} out of total 1192 interlocks in Table~\ref{tab:realtime} for the LASSO\textsubscript{0.8} model) with regard to the full range of forecasting horizon $T_h(\SI{}{\second})$, i.e. from \SI{0.2}{\second} to \SI{59.8}{\second}.  967 interlocks are predicted only \SI{0.2}{\second} before, which takes up 98\% out of all 983 true positives.
But there is one earliest prediction occurring at \SI{51.0}{\second} before the interlock, together with several other earlier captures. This demonstrates the potential for the model to make early predictions. In addition, the relation between forecasting horizon and the definition of real time metrics still needs to be further examined.


\begin{figure}
    \centering
\includegraphics[width=\linewidth]{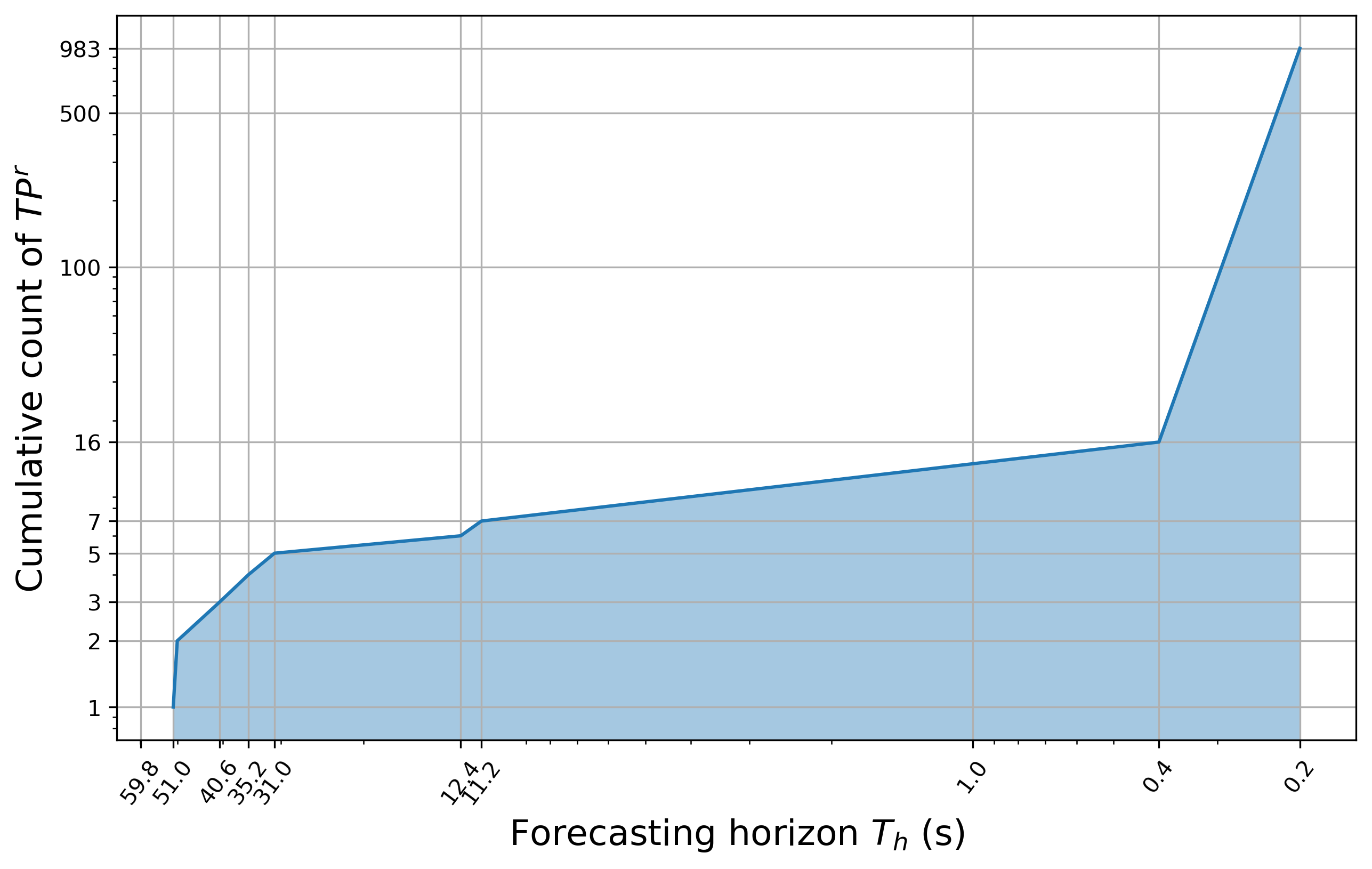}
    \caption{ Cumulative distribution plot of TP\textsuperscript{r} according to forecasting horizon $T_h(s)$. One earliest prediction happens at \SI{51.0}{\second}, and 3 interlocks are predicted longer than \SI{40}{\second}. To emphasize the significant values, both the x and y axes are shown in uneven scale.}
    \label{fig:horizon}
\end{figure}
\section{INSTRUMENTATION}
Although $\text{TPR}^r\approx 85\%$ demonstrates a satisfactory performance  of the LASSO model, the actual ``prediction" of interlock is only available \SI{0.2}{\second} in advance. To perform the recovery operation -- that is, reducing the beam current by 10\% -- in time, we need an instrumentation to realize such change inside \SI{0.2}{\second}. Table~\ref{tab:instrument} lists several possible instrumentation commissioned in various proton facilities at PSI. 

\begin{table}[!hbt]
\centering
\caption{Potential instrumentation for fast adjustment of beam intensity in various proton facilities at PSI}
\begin{tabular}{llr}
\toprule
\textbf{Instrumentation} & \textbf{Facility} & \textbf{Time scale} (\SI{}{\milli
\second}) \\
\midrule
Kicker AVKI & HIPA & 0.005\\
Kicker~\cite{jirousek2003concept} & PROSCAN & 0.05\\
Deflector plates~\cite{carmona18continuous} & Gantry2 & 0.2\\
Beam blocker~\cite{jirousek2003concept} & PROSCAN & 60 \\
Collimator KIP2 & PSI Injector 2 & 66.7/\SI{0.2}{\milli\ampere}\\
\bottomrule
\end{tabular}
\label{tab:instrument}
\end{table}

The AVKI kicker is a fast magnet kicker in the low energy beam line between the Cockroft–Walton pre-accelerator and the \SI{72}{\mega\eV} injector cyclotron~\cite{anicic2005fast}, as shown in the top part of Figure~\ref{fig:hipa}. It is part of the original interlock system --- when an interlock signal is received, it responses in \SI{0.005}{\milli\second} to dump the proton beam.

The Kicker in PROSCAN is an essential part of the \emph{spot scanning} process, where the steered beam is switched on and off between the spots by this fast kicker magnet to apply beam onto the patient as
localized spots~\cite{jirousek2003concept}.

The vertical deflector plates (VD) in Gantry2 aims at accurate (1\% precision) controlling of the beam intensity during fast changes to achieve \emph{continuous line scanning of tumours}. The VDs cut the beam through a set of collimator slits by applying a tunable transversal electric field. The extracted beam current is related to the VD voltage as well as the beam energy~\cite{schatti2014first}.

The beam blockers in PROSCAN are located at the entry of each area, before the degrader, after the above-mentioned kicker, and at the beamline entrance to control the beam intensity. The beam blockers at area entry have a lower response time, yet cannot handle large dose~\cite{jirousek2003concept}. 

Finally, the Collimator KIP2 in HIPA is a movable collimator to maintain the optics and control the injected beam current into the chain of cyclotrons. It is located at the first turn of the \SI{72}{\mega\eV} Injector cyclotron, as shown in Figure~\ref{fig:hipa}~\cite{stetson1992commissioning, tahar2023probing}.

All the listed instrumentations are capable of intercepting the beam within \SI{0.2}{\second}. The time scale listed for kickers, deflector plates and beam blocker are their response time to kick, steer or block the respective beams, while the collimator KIP2~\cite{stetson1992commissioning} moves at a speed of $\sim$\SI{6}{\milli\metre/\second}, which translates to \SI{3}{\milli\ampere/\second} given the default beam current of \SI{2}{\milli\ampere}. Therefore KIP2 needs \SI{66.7}{\milli\second} to change 10\% of the default beam current, i.e. \SI{0.2}{\milli\ampere}.

\section{CONCLUSION and outlook}
We propose an approach based on LASSO logistic regression that tackles the forecasting problem of particle accelerator interruptions as binary classification. Our previous RPCNN model transforms 1-dimensional time series into 2-dimensional images aiming at extracting more refined features, yet its classification power is not as strong as expected. A series of two sample MMD tests show that beam interruptions are more abrupt events than gradual build-ups. Thus the lacking performance of RPCNN might be attributed to inappropriate choices of input data as well as complex model parameters. Based on the MMD test result, a simpler LASSO model is established, and it outperforms RPCNN in both standard classification and custom real-time metrics. The list of possible instrumentation for fast adjustment of beam intensity sheds light on future prospect of integrating the model into real-time operation and finally preventing interlocks upon predictions.

\section{ACKNOWLEDGEMENTS}
We would like to place our special thanks to M\'elissa Zacharias for her work on the RPCNN model and hyper-parameter tunning; to Jaime Coello de Portugal for his work on data collection, and to both of their work on the GUI for the real-time operating system. We would like to thank Davide Reggiani for his support on HIPA and 
the instrumentation. We would like to thank Anastasia Pentina and other colleagues from the Swiss Data Science Center for their insightful collaboration and generous support throughout the research. We also thank Hubert Lutz and Simon Gregor Ebner from PSI for their expert knowledge of the HIPA Archiver and help in data collection. We acknowledge the assistance of Derek Feichtinger and Marc Caubet for their help with the Merlin cluster which enables the computational work of the research.

\appendix*
\section{Detailed derivation about  MMD}\label{appendix}

We start from the RKHS $\mathcal{H}$ and function $f:\mathcal{X} \to \mathbb{R} \in \mathcal{F} \in \mathcal{H}$, where $\mathcal{F}$ is a unit ball. Due to the continuity of the evaluation operator $\delta_x$~
\cite{scholkopf2002learning}
\begin{equation}
    \delta_x: \mathcal{H} \to \mathbb{R}, f \mapsto f(x),
\end{equation}
by the Riesz representation theorem~\cite{reed1980functional}, there is a feature mapping $\phi_x$ i.e.
\begin{equation}\label{eq:kernel}
    \phi_x = k(x, \cdot): \mathcal{X} \to \mathbb{R}
\end{equation}
that satisfies
\begin{equation}
    \langle f, \phi_x\rangle_{\mathcal{H}} = f(x).
\end{equation}
where $\langle\cdot, \cdot\rangle_{\mathcal{H}}$ is the inner product between two functions in $\mathcal{H}$. The kernel $k: \mathcal{X} \times \mathcal{X} \to \mathbb{R}$ in Eq.\eqref{eq:kernel} has one argument fixed at $x$, and the other free argument serves as the input argument for the feature map $\phi_x(\cdot)$. In particular, we have
\begin{equation}
    \langle\phi_{x_0}, \phi_{x_1}\rangle_{\mathcal{H}} = k(x_0,x_1).
\end{equation}

\begin{theorem}[MMD as the distance between mean embeddings in $\mathcal{H}$~\cite{gretton2006kernel}]\label{mmddef}

If mean embeddings $\mu_p$ and $\mu_q$ exist, where
\begin{equation*}
    \mathbf{E}_{x_0\sim p}f = \langle f, \mu_p\rangle_{\mathcal{H}}, \quad
    \mathbf{E}_{x_1\sim q}f = \langle f, \mu_q\rangle_{\mathcal{H}}, \quad \forall f \in \mathcal{F}
\end{equation*}
then
\begin{equation}
    \text{MMD}(\mathcal{F}, p, q) = \|\mu_p - \mu_q\|_{\mathcal{H}}.
\end{equation}
\end{theorem}

\begin{proof}
From Eq.~\eqref{eq:mmd1} we have
\begin{align*}
    \text{MMD}(\mathcal{F}, p, q) &= \sup_{f\in \mathcal{F}}\left(\mathbf{E}_{x_0}[f(x_0)]-\mathbf{E}_{x_1}[f(x_1)]\right) \\
    &= \sup_{\|f\|_ \mathcal{H}\leq 1}\left(\langle f, \mu_p\rangle_{\mathcal{H}}-\langle f, \mu_q\rangle_{\mathcal{H}}\right) \\
    &= \sup_{\|f\|_ \mathcal{H}\leq 1}\left(\langle f, \mu_p-\mu_q\rangle_{\mathcal{H}}\right)\\
    &= \left\langle \frac{\mu_p-\mu_q}{\|\mu_p-\mu_q\|_{\mathcal{H}}}, \mu_p-\mu_q\right\rangle_{\mathcal{H}} \\
    &= \|\mu_p-\mu_q\|_{\mathcal{H}}
\end{align*}
\end{proof}

We know that mean embeddings exist when $k(\cdot, \cdot)$ is measurable and $\mathbf{E}_x\sqrt{k(x,x)}<\infty$~\cite{gretton2012kernel}. On top of those, we now set up conditions for MMD to become a metric, i.e. uniquely equals zero if and only if two distributions are the same.
\begin{theorem}[MMD as a metric]
If $\mathcal{H}$ is a \textbf{universal} RKHS defined on a \textbf{compact} metric space $\mathcal{X}$, and the associated kernel $k(\cdot, \cdot)$ is \textbf{continuous}, then
\begin{equation}
    \text{MMD}(\mathcal{F}, p, q) = 0 \iff p=q
\end{equation}
where $\mathcal{F}$ is a unit ball on $\mathcal{H}$ and $p, q$ are two distributions from $\mathcal{X}$.
\end{theorem}

\begin{proof}
It is clear from the expression of MMD (Eq.~\eqref{eq:mmd1} and Eq.~\eqref{eq:mmd2}) that if $p=q$, MMD$(F, p, q)$ is zero. Therefore we only need to prove the rightward arrow.

Since $\mathcal{H}$ is universal, the kernel $k$ is required to be continuous, and $\mathcal{H}$ is also dense with regard to the $L_{\infty}$ norm in the bounded continuous function space $C(\mathcal{X})$~\cite{micchelli2006universal}. In other words, any function in $C(\mathcal{X})$ can be arbitrarily well-approximated by functions in $\mathcal{H}$. Therefore, for any given $\epsilon > 0$ and any function $g \in C(\mathcal{X})$, there exists a function $h \in \mathcal{H}$ such that
\begin{equation}\label{eq:dense}
    \|g-h\|_{\infty} < \epsilon.
\end{equation}
Then we have
\begin{align*}
    |\mathbf{E}_{x_0} g(x_0) - \mathbf{E}_{x_1} g(x_1)|&\leq |\mathbf{E}_{x_0} g(x_0) - \mathbf{E}_{x_0} h(x_0)|\\
    &+ |\mathbf{E}_{x_0} h(x_0) - \mathbf{E}_{x_0} h(x_1)|\\
    &+ |\mathbf{E}_{x_1} h(x_1) - \mathbf{E}_{x_1} g(x_1)| \\
    &\leq \mathbf{E}_{x_0} |g(x_0)-h(x_0)|\\
    &+|\langle h, \mu_p-\mu_q\rangle|_{\mathcal{H}}\\
    &+\mathbf{E}_{x_1}|h(x_1)-g(x_1)|
\end{align*}
and from Eq.~\eqref{eq:mmd2} we know that $\text{MMD}=0$ implies $\mu_p - \mu_q = 0$; from Eq.~\eqref{eq:dense} we have $\mathbf{E}_{x}|g(x)-h(x)| < \epsilon$. Therefore
\begin{equation}
    |\mathbf{E}_{x_0} g(x_0) - \mathbf{E}_{x_1} g(x_1)| \leq \epsilon + 0 + \epsilon = 2\epsilon
\end{equation}
for all $\epsilon > 0$ and $g\in C(\mathcal{X})$. This indicates that the two distributions or Borel probability measures $p$ and $q$ are equal~\cite{gretton2012kernel, dudley2018real}.
\end{proof}

\bibliography{main}
\end{document}